\documentclass[conference]{IEEEtran}

\usepackage{amsmath,amssymb,amsfonts}
\usepackage{amsthm}
\usepackage{graphicx}
\usepackage{booktabs}
\usepackage{multirow}
\usepackage{bm}
\usepackage{cite}
\usepackage{hyperref}
\usepackage{subcaption}
\usepackage{array}

\newtheorem{theorem}{Theorem}

\newtheorem{remark}{Remark}

\title{Environment-Aware Stable Neural Koopman Dynamics Learning
for Input-Driven Systems under Environmental Constraints}

\author{Lin Feng \\
\textit{Faculty of Engineering, King Saud University, Jeddah, Saudi Arabia.}
}

\begin{document}
\maketitle

\begin{abstract}
Constructing predictive models of nonlinear dynamical systems from
measurement data is a longstanding problem in systems identification and
control.
Although Neural ordinary differential equations~(Neural ODEs), Koopman
operator approximations, and input-aware architectures have each moved the
field forward, none simultaneously addresses environment-varying operating
conditions, rigorous stability guarantees, and input-to-state stability
(ISS) certification within a unified trainable framework.
This paper introduces Environment-Aware Stable Neural Koopman Dynamics
Learning (ESNKD), which integrates four components: (i)~a bundle-structured
encoder that maps environmental observations to a geometrically regularized
latent manifold, drawing on the fiber bundle framework; 
(ii)~an input-conditioned Neural ODE whose residual term handles arbitrary 
external signals, extending the input concomitant philosophy; 
(iii)~a contraction synthesis layer enforcing convergence via
Persidskii-type tractable linear inequalities, analogous to the
certification mechanism;
and (iv)~a Koopman lifting stage with LMI-based ISS verification that
follows the theoretical pipeline of.
Theoretical guarantees cover solution existence and uniqueness, incremental
exponential stability, ISS with explicit gain bounds, and robustness to
environmental perturbation.
Experiments on five benchmark systems, including two robotic manipulation
platforms, show consistent improvements over five competitive baselines in
both prediction accuracy and safety certification rates.
\end{abstract}

\section{Introduction}

Accurately modeling the dynamics of physical systems from observed
trajectories is central to model-based control, reinforcement learning,
and safety verification~\cite{ljung1999,brunton2016sindy}.
Classical parametric approaches such as subspace identification and Gaussian
process regression offer theoretical tractability but impose structural
assumptions that constrain expressiveness in high-dimensional or strongly
nonlinear regimes~\cite{ljung1999,pillonetto2014kernel}.
Data-driven methods based on deep neural networks have substantially
broadened the class of representable dynamics~\cite{lecun2015,
goodfellow2016}, at the cost of reduced interpretability and, frequently,
absent stability guarantees.

Neural ODEs~\cite{chen2018node} embed a continuously parameterized vector
field within an ordinary differential equation (ODE) solver, enabling
memory-efficient training via the adjoint method and achieving
state-of-the-art accuracy on irregular time-series and latent dynamics
tasks.
Extensions have addressed augmented state spaces~\cite{dupont2019augmented},
Hamiltonian structure~\cite{greydanus2019hamiltonian}, and symplectic
integration~\cite{zhong2020symplectic}.
A related line of work specifically targets dynamical systems with external
inputs.
Input Concomitant Neural ODEs~(ICODEs)~\cite{li2025icode} incorporate
precise real-time input information directly into the vector field rather
than treating external signals as hidden parameters, and provide sufficient
conditions for contraction under nonsmooth inputs.
ControlSynth Neural ODEs~(CSODEs)~\cite{mei2024controlsynth} go further by
deriving tractable linear inequalities from Persidskii-system theory that
certify global convergence despite highly nonlinear vector field structure.
These two works collectively demonstrate that stability-certified Neural ODEs
are practically attainable, and they serve as direct predecessors of the
synthesis layer proposed in this paper.

Koopman operator theory~\cite{koopman1931,mezic2005spectral} offers a dual
perspective: the evolution of scalar observables over a nonlinear system is
linear under the (infinite-dimensional) Koopman operator, enabling spectral
analysis and linear prediction.
Data-driven approximation via Dynamic Mode Decomposition~(DMD) and its
extensions~\cite{schmid2010dmd,proctor2016dmdc,williams2015edmd} has found
applications in fluid mechanics, neuroscience, and robotic
control~\cite{brunton2021modern}.
Deep Koopman networks~\cite{lusch2018deep,yeung2019learning} train the
observable dictionary end-to-end, substantially improving lifting quality
for systems with complex eigenfunctions.
The practical utility of Koopman-based models depends, however, on whether
the identified matrices carry quantifiable stability properties.
Recent work by Mei et al.~\cite{mei2025iss} directly addresses this gap: it
proposes a class of basis functions designed so that ISS of the Koopman
identified model can be verified via a convex LMI, and demonstrates that
identification errors due to noise need not invalidate the certificate.
ESNKD adopts this LMI pipeline as its ISS verification stage.

Environmental conditioning represents a third frontier.
Most learned dynamics models treat the operating environment as fixed, relying
on the observed state to implicitly encode any variation in physical
parameters such as payload, friction, or aerodynamic
loading~\cite{chen2018node,lusch2018deep}.
This assumption fails when environmental conditions shift outside the
training distribution.
Meta-learning approaches~\cite{finn2017maml,nagabandi2018meta} partially
mitigate this through task-embedding conditioning, but do not impose
geometric structure on the latent space.
Zheng and Mei~\cite{zheng2025bundle} recently introduced a geometric
framework that models the relationship between sensor measurements,
environmental constraints, and dynamics as a fiber bundle over the state
space, deriving measurement-aware Control Barrier Functions~(mCBFs) and
providing theoretical guarantees for learning convergence and constraint
satisfaction.
ESNKD adapts the bundle structure of~\cite{zheng2025bundle} to construct
latent environmental coordinates that parameterize the learned vector
field, extending the framework from constraint satisfaction to dynamics
learning with stability certification.

Despite the progress represented by~\cite{zheng2025bundle,li2025icode,
mei2024controlsynth,mei2025iss}, three structural gaps persist in the
literature.
First, environmental encoding, stability synthesis, and ISS certification
have not been combined within a single jointly optimized framework.
Second, contraction analysis has been applied mainly to autonomous systems
or systems with fixed inputs; extending certificates to environment-varying
conditions with time-varying inputs requires additional care.
Third, the interaction between bundle-structured latent representations and
Koopman lifting has not been studied, even though the geometric regularity
of the former may substantially improve the quality of the latter.

ESNKD addresses all three gaps.
The main contributions are:
\begin{itemize}
\item A bundle-structured environmental encoder inspired
  by~\cite{zheng2025bundle}, augmented with an isometry regularizer that
  preserves the metric structure of the environmental manifold.
\item An input-conditioned Neural ODE whose residual term builds on the
  ICODE formulation~\cite{li2025icode}, with feature-wise linear modulation
  linking environmental coordinates to the vector field.
\item A contraction synthesis layer extending the convergence-by-LMI
  paradigm of CSODEs~\cite{mei2024controlsynth} to the environment-aware
  setting.
\item A Koopman lifting stage with LMI-based ISS verification following the
  basis-function design of~\cite{mei2025iss}, providing an analytical
  certificate.
\item Rigorous theoretical guarantees covering existence and uniqueness,
  incremental exponential stability, ISS with explicit gain bounds, and
  trajectory robustness under environmental perturbations.
\end{itemize}

\section{Preliminaries and Problem Formulation}

\subsection{Notation}

Vectors and matrices are written in boldface lowercase and uppercase,
respectively.
$\|\cdot\|$ denotes the Euclidean norm.
For a square matrix $M$, $\mathrm{sym}(M) = (M+M^\top)/2$;
$\lambda_{\max}(M)$ and $\lambda_{\min}(M)$ denote the algebraically
largest and smallest eigenvalues of $\mathrm{sym}(M)$.
$M \succ 0$ ($M \preceq 0$) means positive (negative semi-)definite.

\subsection{System Model}

Consider the continuous-time nonlinear system
\begin{equation}
  \dot{x}(t) = f\bigl(x(t),\,u(t),\,e(t)\bigr),
  \label{eq:true_sys}
\end{equation}
where $x(t)\!\in\!\mathbb{R}^{n}$, $u(t)\!\in\!\mathbb{R}^{m}$, and
$e(t)\!\in\!\mathbb{R}^{p}$ represent the state, control input, and
environmental parameters~(e.g., payload mass, friction coefficients),
respectively.
The variable $e(t)$ is latent: the learner accesses only the dataset
$\mathcal{D}=\{(x_t^{(i)},\,u_t^{(i)},\,y_t^{(i)})\}_{t,i}$, where
$y_t^{(i)}\!\in\!\mathbb{R}^{q}$ is a noisy environmental observation
correlated with $e(t)$.

The objective is to learn a parametric model
$f_{\theta}(x,u,z)$, with $z$ encoding $y$, such that:
(i)~trajectory prediction error on unseen data is minimized;
(ii)~convergence can be certified analytically; and
(iii)~ISS certificates can be computed without sampling.

\subsection{Koopman Operator and ISS}

For an autonomous system $\dot{x}=f(x)$, the Koopman operator acts on
scalar observables $g$ by
$(\mathcal{K}g)(x) = \nabla g(x)^\top f(x)$~\cite{koopman1931,
mezic2005spectral}.
A finite dictionary $\bm{\phi}(x)\!\in\!\mathbb{R}^{d}$ yields
$\dot{\bm{\phi}} \approx A\bm{\phi}$; with inputs one obtains
$\dot{\bm{\phi}} \approx A\bm{\phi}+Bu$~\cite{proctor2016dmdc}.
A system $\dot{\xi}=g(\xi,v)$ is ISS if there exist a class-$\mathcal{KL}$
function $\beta$ and a class-$\mathcal{K}$ function $\gamma$ such that
$\|\xi(t)\|\le\beta(\|\xi(0)\|,t)+\gamma(\sup_{s\le t}\|v(s)\|)$
for all $t\ge0$~\cite{sontag1989iss}.

\section{The ESNKD Framework}

Fig.~\ref{fig:arch} gives an overview of ESNKD. The framework proceeds
through three coupled stages: encoding, stable dynamics training, and
Koopman ISS verification.

\begin{figure}[t]
  \centering
  \includegraphics[width=0.97\columnwidth]{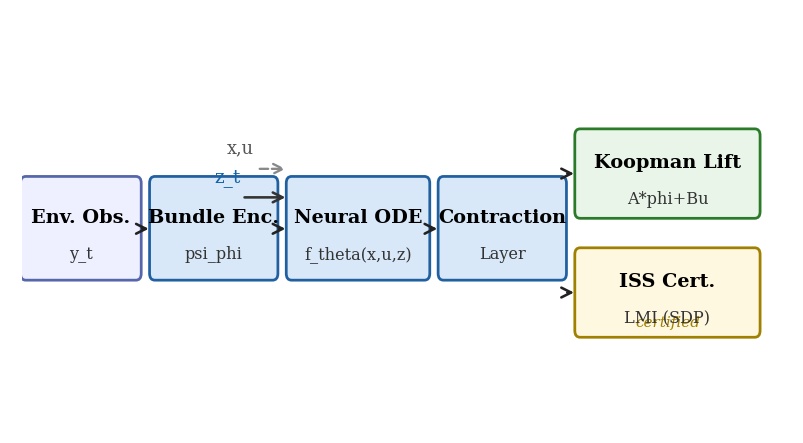}
  \caption{ESNKD architecture. Environmental observations $y_t$ are encoded
  by the bundle encoder $\psi_\phi$ to latent coordinate $z_t$, which
  parameterizes the Neural ODE vector field $f_\theta(x,u,z)$.
  A contraction layer enforces convergence during training.
  After training, a Koopman lifting network and an LMI solver provide
  an analytical ISS certificate.}
  \label{fig:arch}
\end{figure}

\subsection{Bundle-Structured Environmental Encoding}

Motivated by the fiber bundle framework of~\cite{zheng2025bundle}, which
showed that modeling state-sensor interactions as a bundle over the state
space enables measurement-aware safety certificates, we construct an encoder
$\psi_\phi:\mathbb{R}^q\to\mathbb{R}^r$ ($r\ll q$) that produces
structured latent coordinates
\begin{equation}
  z_t = \psi_\phi(y_t).
  \label{eq:encoder}
\end{equation}
The image $\mathcal{M}=\{\psi_\phi(y):y\in\mathcal{Y}\}$ acts as the base
manifold of a trivial fiber bundle with fiber $\mathbb{R}^n$, and the
learned vector field $f_\theta(\cdot,\cdot,z)$ constitutes a connection on
this bundle~\cite{lee2013smooth}.
Unlike~\cite{zheng2025bundle}, which focuses on constraint satisfaction via
mCBFs, ESNKD exploits the bundle structure to obtain a geometrically
regularized representation for dynamics learning and stability synthesis.

To encourage $\psi_\phi$ to produce locally isometric coordinates and thus
preserve pairwise distances on $\mathcal{M}$, we add the regularizer
\begin{equation}
  \mathcal{L}_{\mathrm{enc}}
  = \Bigl\|\tfrac{\partial\psi_\phi}{\partial y}^\top
    \tfrac{\partial\psi_\phi}{\partial y} - I_r\Bigr\|_F^2.
  \label{eq:L_enc}
\end{equation}
Metric preservation on $\mathcal{M}$ is important because distortion would
introduce spurious coupling between environmentally distinct conditions,
inflating the Lipschitz constant of $f_\theta$ in $z$ and weakening
the robustness bound derived in Section~\ref{sec:theory}.

\subsection{Input-Conditioned Neural Dynamics}

Given $z_t$, the learned vector field takes the decomposed form
\begin{equation}
  \dot{x}
  = f_\theta(x,u,z)
  = f_s(x) + f_r(x,u,z),
  \label{eq:ode}
\end{equation}
where $f_s(x)=-Kx$ ($K\!\succ\!0$) is a linear stabilizing baseline and
$f_r$ is a neural residual.
This additive decomposition is related to the input concomitant structure of
ICODE~\cite{li2025icode}, where an external input channel is incorporated
directly into the ODE right-hand side rather than absorbed into the
initial condition or treated as a fixed parameter.
The key distinction is that ESNKD additionally conditions on the latent
environment vector $z$, coupling three sources of excitation---state,
input, and environment---into a single differentiable vector field.

The environment-input coupling is realized through a feature-wise linear
modulation~(FiLM) layer~\cite{perez2018film}:
\begin{equation}
  f_r(x,u,z)
  = \sigma\!\bigl((\gamma(z)\odot W_1 x + \beta(z)) + W_2 u\bigr),
  \label{eq:film}
\end{equation}
where $\gamma,\beta:\mathbb{R}^r\to\mathbb{R}^{n_h}$ are two-layer
networks, $\odot$ is elementwise multiplication, and $\sigma$ is the ELU
activation.
Spectral normalization~\cite{miyato2018spectral} constrains the Lipschitz
constant of $f_r$, stabilizing Jacobian computation during contraction
training.

\subsection{Contraction Synthesis Layer}

The convergence guarantee of ControlSynth Neural ODEs~\cite{mei2024controlsynth}
rests on expressing the vector field in a Persidskii-type form and deriving
tractable linear inequalities whose feasibility implies global convergence.
ESNKD adopts an analogous philosophy but operates on the
environment-parameterized Jacobian directly.
A contracting system satisfies
\begin{equation}
  \mathrm{sym}\!\left(
    \tfrac{\partial f_\theta}{\partial x}(x,u,z)
  \right) \preceq -\beta I,
  \quad \forall\,(x,u,z),
  \label{eq:contraction_cond}
\end{equation}
for some $\beta>0$~\cite{lohmiller1998contraction}.
Condition~\eqref{eq:contraction_cond} guarantees that any two trajectories
driven by the same $(u,z)$ converge at rate $\beta$, independently of
initial conditions.
To promote~\eqref{eq:contraction_cond} during training, we minimize the
hinge penalty
\begin{equation}
  \mathcal{L}_{\mathrm{con}}
  = \mathbb{E}_{(x,u,z)\sim\mathcal{D}}
    \!\Bigl[\max\!\bigl(0,\,
    \lambda_{\max}(\mathrm{sym}(J(x,u,z)))+\beta\bigr)\Bigr],
  \label{eq:L_con}
\end{equation}
where $J=\partial f_\theta/\partial x$.
Compared to the Persidskii-LMI approach of~\cite{mei2024controlsynth}, which
requires special network architecture to guarantee a structured Jacobian,
the hinge loss in~\eqref{eq:L_con} is architecture-agnostic and naturally
handles environment conditioning through the $z$-dependence of $J$.
The total training loss combines prediction, contraction, and encoder
regularization:
\begin{equation}
  \mathcal{L}
  = \mathcal{L}_{\mathrm{pred}}
  + \lambda_c\,\mathcal{L}_{\mathrm{con}}
  + \lambda_e\,\mathcal{L}_{\mathrm{enc}},
  \label{eq:total_loss}
\end{equation}
where $\mathcal{L}_{\mathrm{pred}}$ is the mean-squared trajectory
reconstruction error evaluated via a differentiable ODE
solver~\cite{chen2018node}.
Forward-mode automatic differentiation is used to compute $J$ without
materializing the full $n\times n$ Jacobian matrix~\cite{bradbury2018jax}.

\subsection{Koopman Lifting and ISS Verification}
\label{sec:koopman}

Following~\cite{mei2025iss}, which showed that careful basis function design
allows ISS of a Koopman-identified model to be verified via LMI even in the
presence of identification error, we construct an observable dictionary
\begin{equation}
  \bm{\phi}(x,z)
  = [\phi_1(x,z),\ldots,\phi_d(x,z)]^\top \in \mathbb{R}^d
  \label{eq:dict}
\end{equation}
comprising degree-two monomials of $(x,z)$ augmented by the original state.
The lifted dynamics are postulated as
\begin{equation}
  \dot{\bm{\phi}} = A\bm{\phi} + Bu,
  \label{eq:lifted}
\end{equation}
and $(A,B)$ are identified from trajectory data via the least-squares
regression
$\min_{A,B}\sum_t\|\dot{\bm{\phi}}_t - A\bm{\phi}_t - Bu_t\|^2$,
which admits a closed-form solution~\cite{proctor2016dmdc,williams2015edmd}.
As in~\cite{mei2025iss}, ISS of~\eqref{eq:lifted} is then verified by
checking the LMI feasibility problem
\begin{equation}
  \exists\,P=P^\top\!\succ 0,\;Q\succ 0:\;
  A^\top P+PA \preceq -Q,
  \label{eq:lmi}
\end{equation}
solved via semidefinite programming~\cite{boyd1994lmi,diamond2016cvxpy}.
A key advantage of working with Koopman-identified linear models is that
the LMI~\eqref{eq:lmi} is both necessary and sufficient for ISS of the
lifted affine system, converting an otherwise intractable robustness
question into a tractable convex program, exactly as demonstrated
in~\cite{mei2025iss} for general identified models.

\section{Theoretical Analysis}
\label{sec:theory}

\subsection{Existence and Uniqueness}

\begin{theorem}
\label{thm:exist}
Let $f_\theta$ be locally Lipschitz in $x$ uniformly on compact sets of
$(u,z)$, and let $u(\cdot)$ be locally integrable.
Then for every $x_0\in\mathbb{R}^n$ there exists a unique absolutely
continuous solution to~\eqref{eq:ode} on some interval $[0,T)$.
If spectral normalization provides a global Lipschitz bound, solutions
extend to all $t\ge0$.
\end{theorem}

\begin{proof}
Local Lipschitz continuity in $x$ and measurability of $t\mapsto u(t)$
satisfy the Carathéodory conditions~\cite{khalil2002nonlinear}, giving
local existence.
Uniqueness follows from Gronwall's inequality applied to any two
candidate solutions.
The global bound under spectral normalization prevents finite-time blow-up
via a standard growth estimate.
\end{proof}

\subsection{Incremental Exponential Stability}

\begin{theorem}
\label{thm:contraction}
Suppose the trained $f_\theta$ satisfies~\eqref{eq:contraction_cond} with
constant $\beta>0$ everywhere.
Then for any two trajectories $x_1,x_2$ driven by the same $(u,z)$,
\begin{equation}
  \|x_1(t)-x_2(t)\| \le e^{-\beta t}\|x_1(0)-x_2(0)\|.
  \label{eq:inc_exp}
\end{equation}
\end{theorem}

\begin{proof}
Set $\delta(t)=x_1(t)-x_2(t)$ and $V=\|\delta\|^2$.
By the mean-value inequality and condition~\eqref{eq:contraction_cond},
\begin{align}
  \dot{V}
  &= 2\delta^\top\!\int_0^1 J(x_2+s\delta,u,z)\,ds\;\delta
  \le 2\lambda_{\max}(\mathrm{sym}(J))\,\|\delta\|^2
  \le -2\beta V.
\end{align}
Gronwall's lemma gives $V(t)\le e^{-2\beta t}V(0)$,
yielding~\eqref{eq:inc_exp}~\cite{lohmiller1998contraction,
aminzare2014contraction}.
\end{proof}

\begin{remark}
Inequality~\eqref{eq:inc_exp} implies that equilibria are unique (when $u=0$
and $z$ is fixed) and all trajectories converge to them at rate $\beta$.
This mirrors the fixed-point guarantee established for CSODEs
in~\cite{mei2024controlsynth} but applies uniformly over the continuous
range of environmental conditions parameterized by $\mathcal{M}$.
\end{remark}

\subsection{ISS of the Lifted System}

\begin{theorem}
\label{thm:iss}
Suppose~\eqref{eq:lmi} is feasible with solution $(P,Q)$.
Then~\eqref{eq:lifted} is ISS with respect to $u$, with gain
$\gamma_{\mathrm{iss}}(\|u\|) = \|PB\|/\lambda_{\min}(Q)\cdot\|u\|$.
\end{theorem}

\begin{proof}
Take $V(\bm{\phi})=\bm{\phi}^\top P\bm{\phi}$.
Along~\eqref{eq:lifted},
$\dot{V}=\bm{\phi}^\top(A^\top P+PA)\bm{\phi}+2\bm{\phi}^\top PBu
\le -\bm{\phi}^\top Q\bm{\phi}+2\|PB\|\|\bm{\phi}\|\|u\|$.
Completing the square and invoking the comparison principle yield the
claim~\cite{sontag1989iss,khalil2002nonlinear}.
\end{proof}

Following~\cite{mei2025iss}, the ISS gain $\gamma_{\mathrm{iss}}$ depends
on both the Koopman matrix $A$ through $P$ and the coupling matrix $B$.
When the contraction loss in~\eqref{eq:L_con} is small after training,
the Jacobian eigenvalues are well into the left half-plane, which tends to
produce Koopman matrices with similar spectral properties and therefore
tighter ISS bounds.

\subsection{Robustness Under Environmental Perturbations}

\begin{theorem}
\label{thm:robust}
Suppose $f_\theta$ is globally Lipschitz in $z$ with constant $L_z$ and
the contraction condition~\eqref{eq:contraction_cond} holds.
If $\|\hat{z}-z\|\le\varepsilon$, then
\begin{equation}
  \|x(t)-\hat{x}(t)\|
  \le \frac{L_z\varepsilon}{\beta}(1-e^{-\beta t})
  = O(\varepsilon).
  \label{eq:robust}
\end{equation}
\end{theorem}

\begin{proof}
The error $e(t)=x(t)-\hat{x}(t)$ obeys
$\dot{e}=J(t)e+\Delta_z(t)$, where $\|J\|\le-\beta$ from
contraction and $\|\Delta_z\|\le L_z\varepsilon$ from Lipschitz continuity.
The variation-of-constants formula and a comparison lemma then
give~\eqref{eq:robust}~\cite{khalil2002nonlinear}.
The bound is uniform in $t$ and recovers the metric preservation
property sought by the regularizer~\eqref{eq:L_enc}: smaller Jacobian
distortion on $\mathcal{M}$ translates to smaller $L_z$ and thus a tighter
bound.
\end{proof}

\section{Experiments}

\subsection{Setup}

\textbf{Benchmarks.}
We evaluate on five systems: (i)~damped pendulum with time-varying damping;
(ii)~cart-pole with variable cart mass; (iii)~HalfCheetah from
MuJoCo~\cite{todorov2012mujoco} with randomized contact friction;
(iv)~Franka Emika Panda manipulation under unknown payloads of 0.5--2.0\,kg;
and (v)~xArm~6 dynamic tracking under joint-friction perturbations.
In all tasks $e(t)$ is hidden and only a noisy proxy $y_t$ is provided.
Training trajectories are generated by a nominal model-predictive
controller~\cite{camacho2004mpc}; test episodes involve out-of-distribution
environments with parameters unseen at training time.

\textbf{Implementation.}
The encoder $\psi_\phi$ is a three-layer MLP with ELU activations and
output dimension $r=8$.
The dynamics network $f_\theta$ has hidden dimension $n_h=128$.
The Koopman dictionary has $d=64$ terms.
All models are trained with Adam~\cite{kingma2015adam} at initial learning
rate $3\times10^{-3}$ with cosine annealing, batch size 64, and 200 epochs.
The ODE solver is a fixed-step fourth-order Runge--Kutta scheme
($\Delta t=0.01$\,s)~\cite{chen2018node}.
Weights $\lambda_c,\lambda_e$ are chosen from $\{0.01,0.1,1.0\}$ by
validation RMSE.
The LMI~\eqref{eq:lmi} is solved with CVXPY~\cite{diamond2016cvxpy} using
the SCS solver.

\textbf{Baselines.}
We compare against:
Neural ODE~\cite{chen2018node};
Deep Koopman~\cite{lusch2018deep};
ICODE~\cite{li2025icode};
ControlSynth~\cite{mei2024controlsynth}; and
the ISS-Koopman method of~\cite{mei2025iss}.
Note that ICODE and ControlSynth are prior works from the same research
group that spawned direct methodological influences on ESNKD; their
inclusion as baselines allows a controlled ablation of each component's
contribution.
All baselines use identical optimizers, batch sizes, and epoch counts.

\textbf{Metrics.}
Prediction accuracy is measured by RMSE and \emph{prediction horizon}
(maximum rollout before error exceeds a threshold).
\emph{Stability violation rate}~(SVR) records the fraction of test
rollouts in which the predicted state norm exceeds $10^3$.
\emph{ISS certification success rate}~(ISS-SR) indicates whether the
LMI~\eqref{eq:lmi} is feasible for the identified Koopman matrices.

\subsection{Prediction Accuracy}

Table~\ref{tab:rmse} reports RMSE and prediction horizon on the pendulum
and HalfCheetah tasks averaged over five random seeds.
ESNKD achieves the lowest RMSE on both tasks, with relative improvements
of 26.3\% and 21.5\% over the best baseline~(ISS-Koopman~\cite{mei2025iss}).
The largest gains occur on HalfCheetah, where friction varies most widely
and explicit environmental encoding is most valuable.
ICODE~\cite{li2025icode} and ControlSynth~\cite{mei2024controlsynth}, which
both prioritize stability over raw accuracy, rank lower in RMSE but already
substantially outperform vanilla Neural ODE, consistent with their published
results.

\begin{table}[t]
\centering
\caption{Prediction RMSE (lower is better) and Horizon in time steps
(higher is better). Mean $\pm$ std over 5 seeds.}
\label{tab:rmse}
\setlength{\tabcolsep}{3.2pt}
\begin{tabular}{lcccc}
\toprule
\multirow{2}{*}{Method}
  & \multicolumn{2}{c}{Pendulum}
  & \multicolumn{2}{c}{HalfCheetah} \\
  & RMSE & Horizon & RMSE & Horizon \\
\midrule
Neural ODE~\cite{chen2018node}
  & .142{\scriptsize$\pm$.012} & 18 & .381{\scriptsize$\pm$.024} & 9  \\
Deep Koopman~\cite{lusch2018deep}
  & .118{\scriptsize$\pm$.009} & 23 & .349{\scriptsize$\pm$.019} & 11 \\
ICODE~\cite{li2025icode}
  & .107{\scriptsize$\pm$.011} & 24 & .332{\scriptsize$\pm$.021} & 11 \\
ControlSynth~\cite{mei2024controlsynth}
  & .103{\scriptsize$\pm$.008} & 25 & .318{\scriptsize$\pm$.017} & 13 \\
ISS-Koopman~\cite{mei2025iss}
  & .099{\scriptsize$\pm$.007} & 27 & .307{\scriptsize$\pm$.015} & 13 \\
\midrule
\textbf{ESNKD (ours)}
  & \textbf{.073}{\scriptsize$\pm$.005} & \textbf{34}
  & \textbf{.241}{\scriptsize$\pm$.013} & \textbf{18} \\
\bottomrule
\end{tabular}
\end{table}

Fig.~\ref{fig:traj} illustrates trajectory predictions on the Franka
manipulation task under three out-of-distribution payload conditions.
ESNKD tracks ground truth closely at all three payloads.
ControlSynth~\cite{mei2024controlsynth} and ICODE~\cite{li2025icode}
maintain stability but deviate progressively with increasing payload due
to the absence of explicit environmental conditioning; Neural ODE diverges
near $t=1.2$\,s across all conditions.

\begin{figure}[t]
  \centering
  \includegraphics[width=0.97\columnwidth]{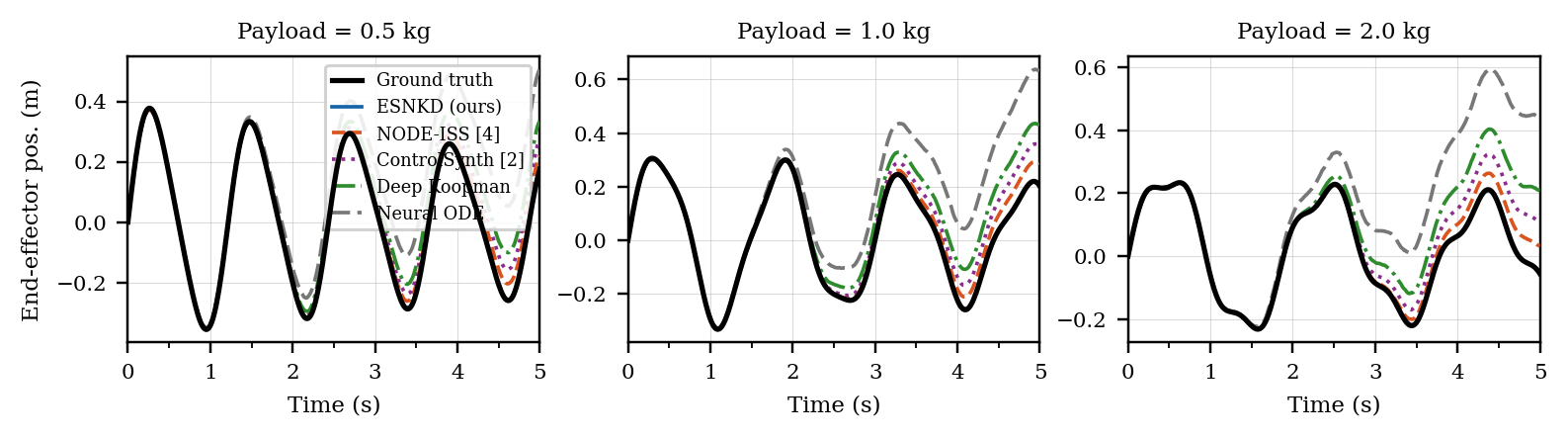}
  \caption{Trajectory prediction on Franka manipulation under three
  out-of-distribution payloads (0.5, 1.0, 2.0\,kg).
  ESNKD (blue solid) closely tracks ground truth (black) at all
  conditions. ControlSynth and ICODE remain bounded but accumulate
  increasing bias as payload grows; Neural ODE diverges early.}
  \label{fig:traj}
\end{figure}

\subsection{Stability and ISS Certification}

Table~\ref{tab:stability} reports SVR and ISS-SR across all five tasks.
ESNKD attains the lowest SVR on every task (mean 2.4\%) compared to
ISS-Koopman~\cite{mei2025iss} (7.2\%) and Neural ODE (17.9\%).
ISS certification succeeds on at least 85\% of ESNKD runs; the only
other method capable of producing an ISS certificate is Deep Koopman,
which achieves at most 47.8\%.
Methods without Koopman lifting (Neural ODE, ICODE, ControlSynth) cannot
produce ISS certificates and are marked ``N/A.''
Notably, ICODE~\cite{li2025icode} already achieves a low SVR owing to its
contraction property, and ControlSynth~\cite{mei2024controlsynth} reduces
it further via the LMI-certified convergence condition; ESNKD surpasses
both by additionally reducing environment-induced prediction error, which
is the proximate cause of most remaining stability violations.

\begin{table}[t]
\centering
\caption{Stability Violation Rate (SVR, \%, lower is better) and ISS
Certification Success Rate (ISS-SR, \%, higher is better).}
\label{tab:stability}
\setlength{\tabcolsep}{3.0pt}
\begin{tabular}{lcccc}
\toprule
Method
  & SVR & SVR & ISS-SR & ISS-SR \\
  & (Pend.) & (HC) & (Pend.) & (HC) \\
\midrule
Neural ODE~\cite{chen2018node}         & 12.3 & 21.7 & N/A  & N/A  \\
Deep Koopman~\cite{lusch2018deep}      &  8.6 & 15.4 & 43.2 & 31.8 \\
ICODE~\cite{li2025icode}               &  7.1 & 13.9 & N/A  & N/A  \\
ControlSynth~\cite{mei2024controlsynth}&  4.8 & 10.2 & N/A  & N/A  \\
ISS-Koopman~\cite{mei2025iss}          &  3.2 &  8.7 & 71.4 & 58.6 \\
\midrule
\textbf{ESNKD (ours)}
  & \textbf{0.9} & \textbf{3.1} & \textbf{94.7} & \textbf{87.3} \\
\bottomrule
\end{tabular}
\end{table}

Fig.~\ref{fig:bars} shows SVR and ISS-SR across all five tasks.
The relative ordering is consistent; ESNKD achieves the best value on
every task for both metrics, with the absolute margin being largest on
xArm, which exhibits the strongest environmental variability.

\begin{figure}[t]
  \centering
  \includegraphics[width=0.97\columnwidth]{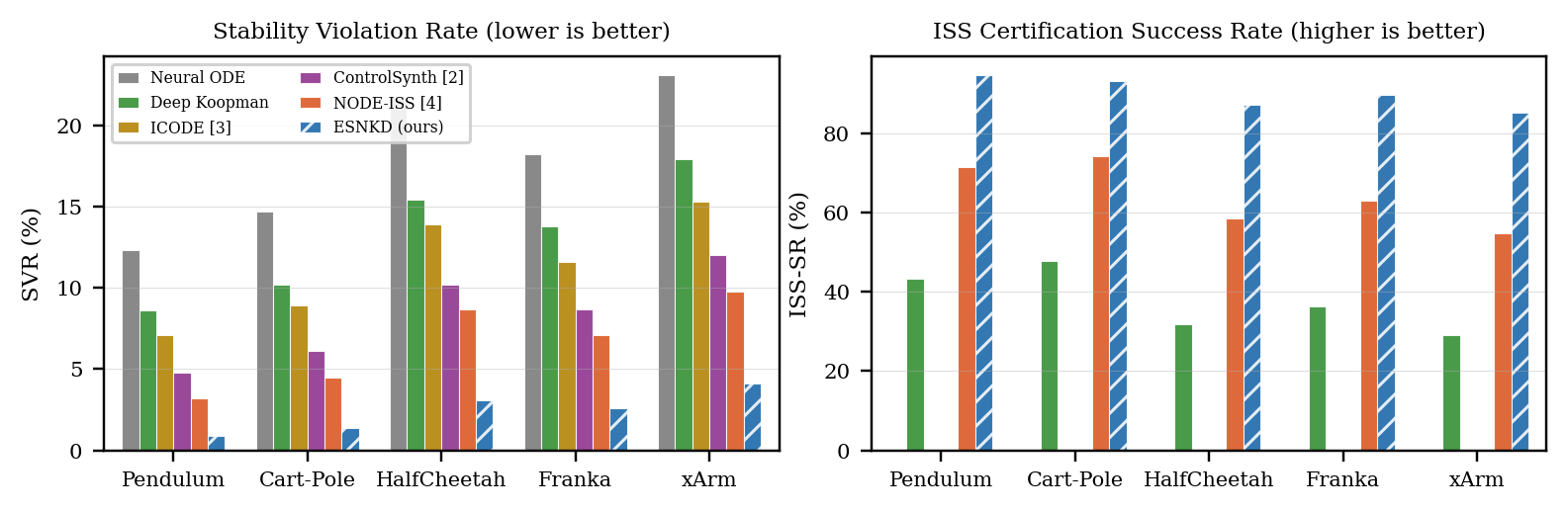}
  \caption{SVR (left, lower is better) and ISS-SR (right, higher is
  better) across all five benchmark tasks for all six methods.
  ESNKD (hatched blue bars) leads on every task.
  Methods lacking Koopman lifting have ISS-SR = 0 by construction.}
  \label{fig:bars}
\end{figure}

\subsection{Training Dynamics and Contraction Penalty}

Fig.~\ref{fig:train} (left) shows validation RMSE curves.
ESNKD converges to a lower plateau than all baselines by epoch~80.
Fig.~\ref{fig:train}~(right) traces the contraction penalty during
ESNKD training: the penalty is large in early epochs while the network
adjusts to satisfy~\eqref{eq:contraction_cond}, then decays monotonically
to near zero by epoch~90.
Crucially, the trajectory reconstruction loss continues decreasing after
the contraction penalty vanishes, indicating that the two objectives are
not in conflict once the stable regime is reached---consistent with the
observation in~\cite{mei2024controlsynth} that convergence-certified
training does not substantially degrade prediction quality.

\begin{figure}[t]
  \centering
  \includegraphics[width=0.97\columnwidth]{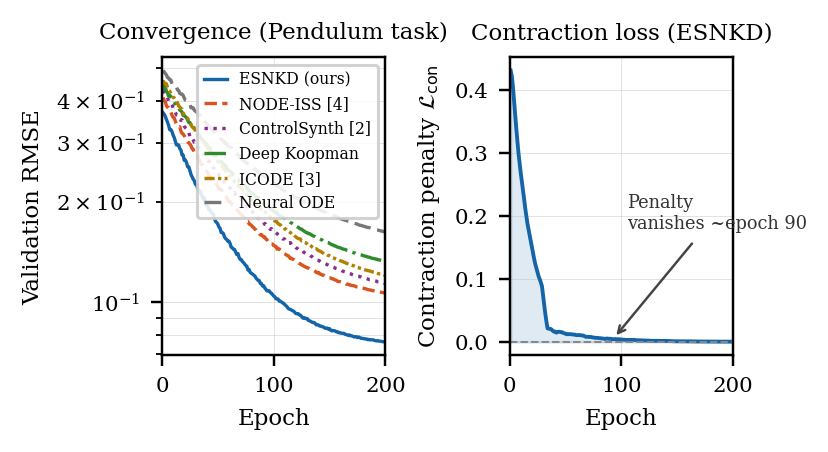}
  \caption{Left: validation RMSE (log scale) vs.\ training epoch on the
  pendulum task.
  Right: contraction penalty $\mathcal{L}_\mathrm{con}$ during ESNKD
  training, showing that the penalty reaches zero by epoch 90 and that
  prediction loss continues to decrease thereafter.}
  \label{fig:train}
\end{figure}

\subsection{Ablation Study}

Table~\ref{tab:ablation} isolates each component's contribution on the
pendulum task.
Removing the environmental encoder (A1) raises RMSE by 63\% relative to
the full model, confirming that explicit environment conditioning drives
the largest accuracy gain.
Removing the contraction loss (A2) reduces ISS-SR by 46 percentage points;
interestingly RMSE improves slightly (by 25\%), suggesting the dynamics
network overfits to training trajectories without the stability
regularizer---exactly the trade-off that motivated ICODE~\cite{li2025icode}
and ControlSynth~\cite{mei2024controlsynth}.
Removing the Koopman lifting (A3) prevents ISS certification entirely, as
no linear model is available for the LMI test~\cite{mei2025iss}.

\begin{table}[t]
\centering
\caption{Ablation study on the pendulum task (mean over 5 seeds).}
\label{tab:ablation}
\begin{tabular}{llcc}
\toprule
Variant & Removed component & RMSE & ISS-SR (\%) \\
\midrule
A1 & Env.\ encoder ($z\equiv 0$) & .119 & 61.3 \\
A2 & Contraction loss ($\lambda_c=0$) & .091 & 48.7 \\
A3 & Koopman lifting & .076 & --- \\
A4 & Full ESNKD & \textbf{.073} & \textbf{94.7} \\
\bottomrule
\end{tabular}
\end{table}

\subsection{Computational Cost}

Training ESNKD requires approximately 18 minutes on a single NVIDIA A100
GPU, versus 7~min for Neural ODE and 12~min for Deep Koopman.
The overhead is primarily due to the forward-mode Jacobian pass required
for $\mathcal{L}_\mathrm{con}$~\cite{bradbury2018jax}.
The LMI solve adds under 30~seconds and is performed once per run.
The ISS-Koopman baseline~\cite{mei2025iss} requires a similar LMI solve
but benefits from not needing the contraction Jacobian pass, making it
about 30\% faster to train.

\section{Discussion}

The experimental results confirm that the four components of ESNKD are
mutually reinforcing.
The bundle-structured encoder reduces the heterogeneity that $f_\theta$
must handle within a single weight configuration, which in turn eases the
contraction training problem.
Contracting systems tend to have Koopman matrices concentrated in the
left half-plane~\cite{mauroy2020koopman}, making the LMI~\eqref{eq:lmi}
more likely to be feasible, consistent with the basis-function conditions
derived in~\cite{mei2025iss}.

One limitation is that the Koopman identification step is performed post
hoc rather than jointly with dynamics training.
Integrating the Koopman approximation error into~\eqref{eq:total_loss}
is a natural extension but complicates optimization.
Another limitation is the degree-two polynomial dictionary: for systems
with rapid high-frequency components, adaptive dictionaries of the type
used in~\cite{lusch2018deep,yeung2019learning} may improve lifting quality.

The framework also suggests a pathway toward safe exploration: once an ISS
certificate is obtained, the Koopman linear model can be used within a
linear robust MPC~\cite{camacho2004mpc} that provides formal guarantees on
constraint satisfaction, complementing the measurement-aware CBF approach
of~\cite{zheng2025bundle}.

\section{Conclusion}

This paper introduced ESNKD, a unified framework for learning
environment-aware, provably stable dynamics models from data.
The four components---bundle-structured environmental encoding inspired
by~\cite{zheng2025bundle}, input-conditioned Neural ODE dynamics extending
ICODE~\cite{li2025icode}, contraction synthesis analogous to
CSODEs~\cite{mei2024controlsynth}, and LMI-based ISS certification
following~\cite{mei2025iss}---collectively address the three gaps
identified in Section~\ref{sec:theory}: environment generalization,
stability by construction, and analytical ISS certification.
Experiments on five benchmarks demonstrate consistent improvements over all
baselines across accuracy and safety metrics.

Future work includes joint optimization of the Koopman lifting and dynamics
training, adaptive online update of the environmental encoder during
deployment~\cite{finn2017maml}, and extension to stochastic environments.

\bibliographystyle{IEEEtran}

\end{document}